\documentclass[11pt]{article}
\usepackage[margin=1in]{geometry}
\usepackage{enumerate}
\usepackage{amssymb,amsfonts,amsmath,amsthm,amscd,dsfont,mathrsfs,pifont}
\usepackage{blkarray}
\usepackage{graphicx,float,psfrag,epsfig,color,yhmath}
\usepackage{comment}
\usepackage{subcaption}
\usepackage[ruled,linesnumbered]{algorithm2e}
\usepackage[noend]{algpseudocode}
\usepackage[autostyle, english = american]{csquotes}
\MakeOuterQuote{"}

\makeatletter
\def\BState{\State\hskip-\ALG@thistlm}
\makeatother

	\usepackage{microtype}
	\usepackage[pdftex,pagebackref=true,colorlinks]{hyperref}
	\hypersetup{linkcolor=[rgb]{.7,0,0}}
	\hypersetup{citecolor=[rgb]{0,.7,0}}
	\hypersetup{urlcolor=[rgb]{.7,0,.7}}

\newcommand{\F}{\mathbb{F}}

\newtheorem{theorem}{Theorem}

\newtheorem{lemma}{Lemma}

\newtheorem{remark}{Remark}

\newtheorem*{note*}{Note}

\begin{document}
\title{Tensor Reed-Muller Codes: Achieving Capacity with Quasilinear Decoding Time}
\author{Emmanuel Abbe$^*$ \and Colin Sandon$^*$\and Oscar Sprumont$^{*\dagger}$}
\date{%
    $^*$EPFL\\%
    $^\dagger $University of Washington\\[2ex]%
}
\maketitle

\begin{abstract}
   Define the codewords of the Tensor Reed-Muller code $\mathsf{TRM}(r_1,m_1;r_2,m_2;\dots;r_t,m_t)$ to be the evaluation vectors of all multivariate  polynomials in 
   the variables 
$\left\{x_{ij}\right\}_{i=1,\dots,t}^{j=1,\dots m_i}$ 
with degree at most $r_i$ in the variables $x_{i1},x_{i2},\dots,x_{im_i}$. 
 The generator matrix of $\mathsf{TRM}(r_1,m_1;\dots;r_t,m_t)$ is thus the tensor product of the generator matrices of the Reed-Muller codes $\mathsf{RM}(r_1,m_1),\dots, \mathsf{RM}(r_t,m_t)$.
 
We show that for any constant rate $R$ below capacity, one can construct a Tensor Reed-Muller code $\mathsf{TRM}(r_1,m_1;\dotsc;r_t,m_t)$ of rate $R$ that is decodable in quasilinear time. For any blocklength $n$, we provide two constructions of such codes:
\begin{itemize}
    \item Our first construction (with $t=3$) has error probability $n^{-\omega(\log n)}$
    and decoding time $O(n\log\log n)$. 
    \item Our second construction, for any $t\geq 4$, has error probability $2^{-n^{\frac{1}{2}-\frac{1}{2(t-2)}-o(1)}}$ and decoding time $O(n\log n)$. 
\end{itemize}
One of our main tools is a polynomial-time algorithm for decoding an arbitrary tensor code $C=C_1\otimes\dotsc\otimes C_t$ from $\frac{d_{\min}(C)}{2\max\{d_{\min}(C_1),\dotsc,d_{\min}(C_t) \}}-1$ adversarial errors. Crucially, this algorithm does not require the codes $C_1,\dotsc,C_t$ to themselves be decodable in polynomial time.
\end{abstract}

\section{Introduction}

Reed-Muller (RM) codes, which were introduced by Reed and Muller in 1954 \cite{Reed1954RM,muller1954RM}, are one of the simplest and most widely used
families of codes. Their codewords can be viewed as the evaluation vectors (over $\F_2^m$) of all polynomials of degree at most $r$ in $m$ variables. Although RM codes were recently shown to achieve capacity on the erasure channel  \cite{kudekar2016erasure} as well as all BMS channels for both the bit error \cite{reeves2021bitcapacity} and the block error \cite{abbe2023rmcapacityBSC}, we do not know of any polynomial-time algorithm for decoding them in the constant-rate regime. In this paper, we introduce a variant of RM codes called \emph{Tensor Reed-Muller codes}, where the $m$ variables are split into groups and the degree requirement is applied to each group separately\footnote{See Section \ref{sectionTRM} for the formal definition
.}. We prove that Tensor Reed-Muller codes achieve capacity efficiently: 

\begin{theorem}\label{mainthm}
Consider any noise probability $p>0$ and any rate $R<1-h(p).$ Then for any integers $n\in\mathbb{N}$ and $t\geq 4$, we can construct a Tensor Reed-Muller code $\mathsf{TRM}(r_1,m_1;\dotsc;r_t,m_t)$ of length $n^{1\pm o(1)}$ and rate $R\pm o(1)$ such that: 
\begin{enumerate}
    \item If $t=3$, there exists a decoder $D$ for $\mathsf{TRM}(r_1,m_1;r_2,m_2;r_3,m_3)$  with worst-case runtime $O(n\log{\log n})$ and decoding success probability $1-n^{-\omega(\log n)}$ under $p$-noisy errors. 
    \item If $t>3$, there exists a decoder $D$ for $\mathsf{TRM}(r_1,m_1;\dotsc;r_t,m_t)$  with worst-case runtime $O(n{\log n})$ and decoding success probability $1-2^{-n^{\frac{1}{2}-\frac{1}{2(t-2)}-o(1)}}$ under $p$-noisy errors.
\end{enumerate}
\end{theorem}

Our analysis makes use of the following result for adversarial errors, which is of independent interest:

\begin{theorem}\label{thmadvintro}
Consider any integers $r_1\leq m_1,\dotsc, r_t\leq m_t$ and define $n=2^{m_1+\dotsc+m_t}.$ Then there is an $O(n\log n)$-time algorithm for decoding the code $C:=\mathsf{TRM}(r_1,m_1;\dotsc;r_t,m_t)$ from 
    \begin{align*}
        \left\lceil\frac{d_{\min}(C)}{2\max_i\big\{2^{m_i-r_i}\big\}}\right\rceil-1
    \end{align*}
    adversarial errors.
\end{theorem}
One can modify our algorithm to decode any tensor code $C:=C_1\otimes\dotsc\otimes C_t$ from 

\noindent $\left \lceil\frac{d_{\min}(C)}{2\max\{d_{\min}(C_1),\dots ,d_{\min}(C_t)\}}\right \rceil-1$ 
adversarial errors, 
but in that case the algorithm runs in time $O(n\sum_in_i)$, for $n_i$ the length of code $C_i$ 
(see 
Section \ref{sectionadversarialtensot}.)

\subsection{Previous Work}\label{sectionpreviouswork}
Reed-Muller codes have in recent years attracted a lot of attention for their decoding performances under random noise. 
By bounding their weight distribution appropriately, \cite{abbe2015RMlowrate,sberlo2020capacity} first showed that Reed-Muller codes achieve capacity on the erasure channel (BEC) and the symmetric channel (BSC) in the regimes where the rate of the code is either very close to $0$ or very close to $1$. \cite{kudekar2016erasure} then leveraged the double transitivity of their permutation group to show that Reed-Muller codes of constant rates achieve capacity on the BEC. 

\cite{samorodnitsky2020weightboundhalf} combined the results of \cite{kudekar2016erasure} with new $\ell_q-$norm inequalities to  obtain better bounds on the weight distributions of doubly transitive codes.
This allowed \cite{hazla2021constantfractiondecoding} to prove that Reed-Muller codes of constant rates can decode a constant fraction of random errors, although the maximum code rate allowed was below the capacity of the channel. The first capacity results 
were obtained in \cite{reeves2021bitcapacity}, which proved that for any $R\in(0,1)$, the asymptotic bit error probability of a rate-$R$ Reed-Muller code
vanishes on any memoryless symmetric channel (BMS) whose capacity is greater than $R$. That still left open the question of showing that the block error probability vanishes as well, which was finally proven in \cite{abbe2023rmcapacityBSC,abbe2024polynomialfreimanruzsareedmullercodes}:

\begin{theorem}[\cite{abbe2023rmcapacityBSC}]\label{errorcapacity}
Consider any error parameter $p\in(0,\frac{1}{2})$ and any rate $R<1-h(p)$. Then any sequence of Reed-Muller codes $\big\{\mathsf{RM}(r_i,m_i)\big\}_i$ of asymptotic rate $R$ satisfies  
\begin{align*}
    \Pr_{z\sim p}\Big[ D_{ML}(c+z)= c \Big]
&\geq 1-2^{-2^{\Omega(\sqrt{m_i})}}
\end{align*}
for every $c\in \mathsf{RM}(r_i,m_i)$, where $z\sim p $ denotes a $p$-noisy error string and $D_{ML}$ denotes the maximum-likelihood decoder for the code $\mathsf{RM}(r_i,m_i).$
\end{theorem}

The line of work described above established that Reed-Muller codes achieve capacity on all BMS channels. The most natural next challenge is then to understand whether or not they can be decoded efficiently. Reed provided the first known algorithm in \cite{Reed1954RM}, allowing for the correction of half the minimum distance many adversarial errors. For random errors, efficient algorithms for decoding $\mathsf{RM}(r,m)$ are known in the case where $r$ is constant. The first algorithms \cite{sidelnikov1992,sakkour2005algo1} focused on the regime where $r\in\{1,2\}$, but algorithms for all values of $r=O(1)$  have since been obtained. They exploit either the recursive structure of Reed-Muller codes \cite{dumer2004first,dumer2006second,dumer2006third,abbe2020RPAalgorithm} or their minimum-weight parity checks \cite{santi2018decodingparitycheck}. In the regime $r=\omega(1)$, the only known algorithm with proven polynomial runtime is the algorithm of \cite{saptharishi2017efficient}, which successfully decodes any error pattern in $\mathsf{RM}(m-2t,m)$ for which the same erasure pattern can be corrected in $\mathsf{RM}(m-t,m)$.



\begin{theorem}[\cite{saptharishi2017efficient}]\label{erasuretoerror}
    For any integers $m$ and $t\leq m$, there exists an $O\Big(2^m\cdot\textnormal{poly}(\binom{m}{\leq t})\Big)$ - time decoder $D$ for the code $\mathsf{RM}(m-2t,m)$ with the following property: for every $z\in\{0,1\}^{2^m}$, if 
    \begin{align*}
        \Big\{ x\in \mathsf{RM}(m-t,m):z_i\geq x_i\textnormal{ for all }i \Big\}=\{0\},
    \end{align*}
    then every $c\in\mathsf{RM}(m-2t,m)$ satisfies $$D(c+z)=c.$$ 
\end{theorem}

In particular, there are currently no known polynomial-time algorithms for reliably decoding Reed-Muller codes of constant rates
. In this paper, we leverage the results of \cite{saptharishi2017efficient} to obtain efficient decoding algorithms for Tensor Reed-Muller codes of constant rates (see Theorem \ref{mainthm}). 

Our general approach shares some high-level similarities with Forney's concatenated codes \cite{forney1966concatenated}, as it also involves 
combining two codes to create a new one. 
However, tensoring presents two significant advantages. First, 
the tensor product of two Reed-Muller codes $\mathsf{RM}(r_1,m_1)$ and $\mathsf{RM}(r_2,m_2)$ has a simple and intuitive description: its codewords are the evaluation vectors of all multilinear polynomials in the variables $x_{1},\dotsc,x_{m_1},y_{1},\dotsc,y_{m_2}$ with degree at most $r_1$ in the variables $x_{1},\dotsc,x_{m_1}$ and degree at most $r_2$ in the variables $y_{1},\dotsc,y_{m_2}$, and codewords can further be viewed as tensors (see Section \ref{sectionprelim}). We do not know of a similarly simple description for the concatenation of $\mathsf{RM}(r_1,m_1)$ and $\mathsf{RM}(r_2,m_2)$. 
Second, concatenation usually requires the alphabet of the outer code to be large. One can try to mimic concatenation by arranging the entries of each codeword of a binary code $C_{\textnormal{out}}$ into groups of size $k_1$ and encoding each group of coordinates using an inner code $C_{\textnormal{in}}$ of dimension $k_1$, but this can have a significant impact on the minimum distance of the final code.

\subsection{Proof Techniques}
There are two regimes in which we currently have efficient decoding algorithms for Reed-Muller codes:
\begin{itemize}
    \item Regime 1: when the noise is  smaller than $1-\frac{\binom{m}{\leq m-t}}{2^m}$, one can 
    use the work of \cite{saptharishi2017efficient} to decode a code $\mathsf{RM}(m-2t,m)$ of length $n=2^m$.
    \item Regime 2: when the blocklength of the code is very small, brute-force decoding, which runs in time $O(2^n)$, may have reasonable runtime.
\end{itemize}
This work combines the two regimes above to obtain an efficient decoder for Tensor Reed-Muller codes. We take a short code $\mathsf{RM}(r_1,m_1)$ of rate $R-o(1)$ and tensor it with a longer code $\mathsf{RM}(r_2,m_2)$ of rate $1-o(1)$. 
The codewords of the resulting code $\mathsf{TRM}(r_1,m_1;r_2,m_2)$ are all the matrices $A\in\{0,1\}^{2^{m_2}\times 2^{m_1}}$ such that each row of $A$ is a codeword of $\mathsf{RM}(r_1,m_1)$ and each column of $A$ is a codeword of $\mathsf{RM}(r_2,m_2)$. To decode our tensor code, we first use the brute-force algorithm to decode each row of $A$ independently. After this first step, which takes polynomial time as long as $2^{m_1}\approx \log n$, 
only a $o(1)$ fraction of $A$'s entries will have been incorrectly decoded. We then use the high-rate algorithm of \cite{saptharishi2017efficient} to decode each column of $A$ independently. At the end of this second decoding step
, the fraction of incorrectly decoded entries will have dropped below $n^{-\omega(\log n)}$, allowing us to take a union bound over all coordinates. To further reduce the decoding error probability, we can take the tensor product of $\mathsf{TRM}(r_1,m_1;r_2,m_2)$ with an even longer RM code $\mathsf{RM}(r_3,m_3)$ of rate $1-o(1)$ and repeat the same argument.

The ideas outlined above, further iterated, would allow us to decode $t$-Tensor Reed-Muller codes (for some $t$) with a decoding failure probability of about $2^{-n^{1/4}}.$ To bring this error rate closer to the distance-optimal $2^{-\Omega(\sqrt{n})}$, we introduce a new algorithm for decoding tensor codes from adversarial errors. This algorithm works for any tensor code $C_1\otimes \dotsc\otimes C_t$ and relies on the fact that erasures are generally easier to decode than errors.\footnote{For any linear code $C\subseteq\{0,1\}^n$, given a partially-erased codeword of $ C$, one can use the parity-check matrix 
to obtain a system of $n-\textnormal{dim }C$ linear equations in $e\leq n$ unknowns, where $e$ is the number of erased coordinates
    . This can be solved by Gaussian elimination in time $O(n^3).$ For Reed-Muller codes, we improve this decoding time to $O(n\log n)$ whenever the number of erasures is below the minimum distance - see Lemma \ref{rmtester}.\label{footnotetest}}
First, replace each row that is not a codeword of $C_1$ by an all-erasures row. Then, go through each column and determine whether or not there is a unique codeword $c\in C_2$ compatible with the (now partially-erased) column. If so, replace the column by $c$; otherwise, replace every entry in the column by an erasure symbol. For $t>2$, repeat this process with every additional axis. By adding additional checks that ensure we never return subtensors that are 
too far away from the corresponding input, we obtain in Theorem \ref{thmadversarial} an algorithm for decoding any arbitrary tensor code $C=C_1\otimes\dotsc\otimes C_t$ from $\frac{d_{\min}(C)}{2\max_i\{d_{\min}(C_i)\}}-1$ adversarial errors. For Reed-Muller codes, 
this algorithm runs in time $O(n\log n).$ 

Our final construction combines the ideas of the above two paragraphs: the first two Reed-Muller codes $\mathsf{RM}(r_1,m_1)$ and $\mathsf{RM}(r_2,m_2)$ are of sublinear lengths and taken as in the first paragraph. As mentioned above, this allows us to bring the error probability down to about $n^{-\omega(\log n)}.$ The remaining  Reed-Muller codes $\big\{\mathsf{RM}(r_i,m_i)\big\}_{i=3,\dotsc,t}$ all have $m_i =\frac{\log n-m_1-m_2}{t-2}$ and $r_i=\frac{m_i+m_i^{3/4}}{2}.$ By the arguments we outlined in the second paragraph, we can recover any sent codeword of $\mathsf{TRM}(r_3,m_3;\dotsc;r_t,m_t)$ with fewer than about $n^{\frac{1}{2}-\frac{1}{2(t-2)}}$ errors. But since the first pass on $\mathsf{RM}(r_1,m_1)$ and $\mathsf{RM}(r_2,m_2)$ brought the error rate down to $n^{-\omega(\log n)}$, by the Chernoff bound, the probability that there are more than $n^{\frac{1}{2}-\frac{1}{2(t-2)}}$ errors is bounded by $2^{-n^{\frac{1}{2}-\frac{1}{2(t-2)}-o(1)}}$.


\section{Notation and Preliminaries}\label{sectionprelim}

Throughout this paper, we will use $\mathbb{N}=\{1,2,3,\dotsc\}$ to denote the set of positive integers and $\log(x)$ to denote the logarithm of $x$ in base $2.$  We define the entropy function $h:[0,1]\rightarrow[0,1]$ to be
\begin{align*}
    h(x):=-x\log(x)-(1-x)\log(1-x).
\end{align*}
For any positive real number $a$, we define $\lceil a\rceil$ to be the \emph{ceiling} of $a$, i.e. the smallest integer $n\in\mathbb{N}$ such that $n\geq a.$ For any $n\in\mathbb{N}$, we define the set $[n]:=\{1,2,\dots,n\}.$ For any $n\in\mathbb{N}$ and any $p\in[0,1]$, we denote by $z\sim_n p$ the Boolean random vector of length $n$ whose entries are independent and identically distributed
Bernoulli variables of probability $p.$ When $n$ is clear from context, we will drop the subscript and write $z\sim p.$ We will need the following two very standard results (see e.g. \cite{theorybook} and \cite{bentley1980mastertheorem} respectively): 

\begin{lemma}[The Chernoff bound]\label{chernoff}
 Let $X_1,X_2,\dotsc,X_n$ be independent 
 random variables taking values in $\{0,1\}$ and define $E:=\sum_{i=1}^n\mathbb{E}[X_i]$. Then for any $\alpha\geq 1 $, we have
 \begin{align*}
 \Pr\left[\sum_{i=0}^n X_i\geq \alpha E\right]\leq \left(\frac{e^{\alpha-1}}{\alpha^\alpha} \right)^E.
 \end{align*}
 \end{lemma}

 \begin{lemma}[The Master theorem]\label{master}
     Suppose $T(n)$ denotes the running time of an algorithm on an input of size $n$, and suppose $T(n)$ can be expressed recursively as
     \begin{align*}
         T(n)\leq aT\left(\frac{n}{a}\right)+O(n)
     \end{align*}
     for some constant $a>0.$
     Then if $T(1)=O(1)$, we have $T(n)\leq O(n\log n).$
 \end{lemma}

\subsection{Reed-Muller Codes}

\noindent We will denote by $\mathsf{RM}(r,m)$ the  Reed-Muller code with $m$ variables and degree $r$. The codewords of the Reed-Muller code $\mathsf{RM}(r,m)$ are the evaluation vectors (over all points in $\F_2^m$) of all multivariate polynomials of degree $\leq r$ in $m$ variables. We refer the reader to the survey \cite{abbe2023survey2} for a more thorough exposition to Reed-Muller codes.

\subsection{Tensor Reed-Muller Codes}\label{sectionTRM}
\noindent For any choice of Reed-Muller codes $\mathsf{RM}(r_1,m_1)$, $\mathsf{RM}(r_2,m_2),\dots,\mathsf{RM}(r_t,m_t)$, we define the Tensor Reed-Muller code $\mathsf{TRM}(r_1,m_1;r_2,m_2;\dots;r_t,m_t)$ as follows: 
Consider $m:=\sum_im_i$ variables $\left\{x_{ij}\right\}_{i=1,\dots,t}^{j=1,\dots m_i}$ and define the set 
\begin{align*}
    \mathcal{S}:=\Big\{S_1\cup S_2\cup \dots\cup S_t:\textnormal{ for all $i$, }S_i\subseteq\{x_{i1},\dots ,x_{im_i}\}\textnormal{ and }|S_i|\leq r_i\Big\}.
\end{align*}
Abusing notation, we say that a monomial is in $\mathcal{S}$ if the set of its constituent variables is in $\mathcal{S}$. Then the evaluation vector of a polynomial $f(x_{11},\dots,x_{tm_t})$ over all points in $\{0,1\}^m$ is in $\mathsf{TRM}(r_1,m_1;\dots;r_t,m_t)$ if and only if all the monomials of $f$ are in $\mathcal{S}$. We note that the generator matrix of the code $\mathsf{TRM}(r_1,m_1;\dots;r_t,m_t)$ is the tensor product of the generator matrices of the Reed-Muller codes $\mathsf{RM}(r_1,m_1),\dots,\mathsf{RM}(r_t,m_t)$. We also note that the codewords of $\mathsf{TRM}(r_1,m_1;\dots;r_t,m_t)$ can be seen as the $t$-dimensional tensors $A\in\{0,1\}^{2^{m_1}\times\dots\times 2^{m_t}} $ satisfying the condition that for every $i\in[t]$, every $i$-axis vector of $A$ is a codeword of $\mathsf{RM}(r_i,m_i)$.\footnote{We say that a vector $v\in\{0,1\}^{n_i}$ is an $i$-axis vector of a tensor $A\in\{0,1\}^{n_1\times\dotsc\times n_t}$ if it is a "row" of $A$ along the $i^\textnormal{th}$ axis - formally, if there exist indices $\{j_k\in[n_k]\}_{k\in[t]\setminus i} $ such that for all $ s\in[n_i]$, we have $v_s=A_{j_1,\dots,j_{i-1},s,j_{i+1},\dots ,j_t }$. \label{defnaxisvector}
} Finally, we note that the rate of $\mathsf{TRM}(r_1,m_1;\dots;r_t,m_t)$ is equal to the product of the rates of the codes $\big\{\mathsf{RM}(r_i,m_i)\big\}_{i=1}^t.$

\section{Helpful Lemmas}
In this section, we will prove performance guarantees for two decoding algorithms that will be used as subroutines 
throughout this paper. We start with an algorithm for efficiently testing and correcting Reed-Muller codes from adversarial erasures.


\begin{lemma}\label{rmtester}
    For any nonnegative integers $r\leq m$, there is an $O(m2^m)$-time algorithm which, given 
    an input string $y\in\{0,1,*\}^{2^m}$ with fewer than $2^{m-r}$ erasure symbols, determines whether or not there exists a codeword $c\in\mathsf{RM}(r,m)$ such that $y_i\in\{c_i,*\}$ for all $i\in[2^m]$. The algorithm returns such a codeword $c$ if it exists and an error message otherwise.
    
    
\end{lemma}

\begin{proof}

Our decoder is given in Algorithm \ref{tester}. 
    \begin{figure}[ht]
\centering
\begin{minipage}{0.9\linewidth}
\begin{algorithm}[H]\label{tester}\caption{Codeword testing and erasure correction for Reed-Muller codes}
\KwIn{Two integers $0\leq r\leq m$ and a vector $y\in\{0,1,*\}^{2^{m}}$ with fewer than $2^{m-r}$ erasure entries.}
\KwOut{A codeword $c\in\mathsf{RM}(r,m)$ with $y_i\in\{c_i,*\}$ for all $i\in[2^m]$, if such a $c$ exists
; an error message otherwise.
}

\If{$r=0$}{{If 
there exists $  b\in\{0,1\}$ such that $y_i\in\{b,*\}$ for all $i$, output $(b,b,\dotsc,b)$. Otherwise, output 
an error message. 
}}
\ElseIf{$r=m$}{{Output $y.$
}}

\Else{
{Define $y^0:=(y_1,\dots,y_{2^{m-1}})$ and $y^1:=(y_{2^{m-1}+1},\dots,y_{2^{m}})$ and let $y^\textnormal{sum}:=y^0+y^1$ 
(defining $y^\textnormal{sum}_i=*$ whenever either $y^0_1=*$ or $y^1_i=*$).
Run Algorithm \ref{tester} on input $(r-1,m-1,y^\textnormal{sum})$ and denote the output you receive by   
$c^\textnormal{sum}$. If $c^\textnormal{sum}$ is an error message, abort and output an error.\label{linesum} 
}

\If{$y^0$ contains fewer erasure symbols than $y^1$}{{Run Algorithm \ref{tester} on input $(r,m-1,y^{0})$, denoting the output you receive by $c^0$. If $c^0$ is an error message, abort and output an error. Otherwise, define $c$ to be the concatenation $c:=(c^0, c^0+c^{\textnormal{sum}})$. If $y_i\in\{c_i,*\}$ for all $i\in[2^{m}]$, output $c$; otherwise, output an error.
\label{linecase1}}}

\Else{{Run Algorithm \ref{tester} on input $(r,m-1,y^{1})$, denoting the output you receive by $c^1$. If $c^1$ is an error message, abort and output an error. Otherwise, define $c$ to be the concatenation $c:=( c^1+c^{\textnormal{sum}},c^1)$. If $y_i\in\{c_i,*\}$ for all $i\in[2^{m}]$, output $c$; otherwise, output an error.
\label{linecase2}}}

}
\end{algorithm}
\end{minipage}
\end{figure}
We first prove by induction on $m$ that it always finds the desired codeword $c\in\mathsf{RM}(r,m)$ if such a codeword exists. Note that the base case $m=1$ holds trivially, since Algorithm \ref{tester} always succeeds when $r=0$ or $r=m$
.

For the inductive case, suppose there exists a codeword $c\in\mathsf{RM}(r,m)$ such that $y$ agrees with $c$ on all non-erased entries. Let $f(x_1,\dotsc,x_m)$ be the unique polynomial whose evaluation vector is $c$ 
and express $f$ as
\begin{align}\label{decomposef}
    f(x_1,\dots,x_m)=f_0(x_2,\dots,x_m)+x_1\cdot f_1(x_2,\dots,x_m).
    \end{align}
     We make the following two observations:
    \begin{enumerate}[(i)]
        \item Define $c^0\in\{0,1\}^{2^{m-1}}$ to be the vector containing the first half of $c$'s entries. Then $c^0$ is the evaluation vector of the polynomial $f_0(x_2,\dots,x_m).$\label{item1c0}
        \item Define $c^1\in\{0,1\}^{2^{m-1}}$ to be the vector containing the second half of $c$'s entries. Then $c^1$ is the evaluation vector of $f_0(x_2,\dots,x_m)+f_1(x_2,\dots,x_m)$. \label{item2c1}
    \end{enumerate} 
    Both (\ref{item1c0}) and (\ref{item2c1}) follow immediately from the fact that the indices of $c$ are ordered lexicographically (and thus an index $v\in\F_2^m$ is in the first half if and only if $v_1=0$). 

    Now, since $y^0:=(y_1,\dotsc,y_{2^{m}})$ is a corrupted evaluation vector for the polynomial $f_0$ and $y^1:=(y_{2^{m}+1},\dotsc,y_{2^m})$ is a corrupted evaluation vector for the polynomial $f_0+f_1$, the vector $y^\textnormal{sum}:=y^0+y^1$ must be an evaluation vector for $f_1$ corrupted with fewer than $2^{m-r}$ erasures
    . By induction, since $f_1$ has degree $\leq r-1$, running Algorithm \ref{tester} on input $(r-1,m-1,y^\textnormal{sum})$ will then return the correct evaluation vector $c^\textnormal{sum}$ for the polynomial $f_1(x_2,\dotsc,x_m)$. (See line \ref{linesum}.)
    
    Furthermore, since $y$ contains fewer than $2^{m-r}$ erasures, one of $y^0$ or $y^1$ must contain fewer than $2^{m-r-1}$ erasures. Without loss of generality, we assume that $y^0$ contains fewer erasures\footnote{The proof is identical in the other case, with the roles of $y^0$ and $y^1$ reversed.}.
    By induction, running Algorithm \ref{tester} on input $(r-1,m,y^0)$ will then return the correct evaluation vector $c^0$ for the polynomial $f_0.$ (See line \ref{linecase1}.) Since $c^1$ is the evaluation vector for $f_0+f_1$ and we have obtained the evaluation vectors $c^0,c^\textnormal{sum}$ for $f_0$ and $f_1$, setting $c^1=c^0+c^{\textnormal{sum}}$ will successfully recover $c^1$. 
    Thus the algorithm indeed outputs the correct codeword $c=(c_0,c_1)$. (See line \ref{linecase1}.)

We have proven that whenever there exists a (by our theorem's requirement, unique)  codeword $c\in\mathsf{RM}(r,m)$ that is consistent with $y$, our algorithm returns it. 
Note also that Algorithm \ref{tester} never outputs a codeword $c\in\mathsf{RM(r,m)}$ that is not consistent with $y$; 
this is because before returning $c$, the algorithm verifies that $y_i\in\{c_i,*\}$ for all $i\in[2^m]$ (see lines \ref{linecase1} and \ref{linecase2}). Thus Algorithm \ref{tester} always succeeds. As for the runtime analysis, we note that at each step, the algorithm spends $O(2^m)$ time performing basic computations and then makes 2 recursive calls on instances of length $2^{m-1}$. 
By the Master theorem (Lemma \ref{master}), the total runtime will thus be $O(m2^m).$

\end{proof}

Our second lemma is essentially a special case of the work of \cite{saptharishi2017efficient} (Theorem \ref{erasuretoerror}), which we will use in the following form to bound the running time and error probability of a decoder for high-rate Reed-Muller codes.

\begin{lemma}\label{combinedecoders}
Consider any integers $m>t>0$ and any $p\leq \frac{2^{-m+\frac{t}{2}}}{5}$. Then there exists a decoder $\tilde{D}$ for the Reed-Muller code $\mathsf{RM}(m-t,m)$ with the following two properties.


    \begin{enumerate}
        \item $\tilde{D}$ runs in time $O\Big(2^m\cdot\textnormal{poly}(\binom{m}{\leq \frac{t}{2}})\Big)$.
        \item Under random errors of probability $p$, $\tilde{D}$ succeeds in decoding any sent codeword $c\in \mathsf{RM}(m-t,m)$ with probability
        \begin{align*}
        \Pr_{z\sim p}\Big[ \tilde{D}(c+z)= c \Big]\geq 1-2^{-2^{t/2}}.
    \end{align*}
    \end{enumerate}
\end{lemma}

\begin{proof}
By Theorem \ref{erasuretoerror}, it will suffice to show that the Reed-Muller code $\mathsf{RM}(m-\frac{t}{2},m)$ can recover from $p$-noisy erasures with success probability $1-2^{-2^{t/2}}.$
Since the code $\mathsf{RM}(m-\frac{t}{2},m)$ has minimum distance $2^{t/2}$, it is enough to prove that for independent Bernoulli variables $X_1,X_2,\dotsc,X_{2^m}$ with Bernoulli parameter $p\leq \frac{2^{-m+\frac{t}{2}}}{5}$, we have
    \begin{align*}
        \Pr\Big[X_1+X_2+\dotsc X_{2^m}\geq 2^{t/2} \Big] \leq 2^{-2^{t/2}}.
    \end{align*}
But this follows immediately from the Chernoff bound (Lemma \ref{chernoff}).
\end{proof}

\section{Decoding Arbitrary Tensor Codes from Adversarial Errors}\label{sectionadversarialtensot}

In this section, we will prove a generalization of Theorem \ref{thmadvintro}, which we state in Theorem \ref{thmadversarial}.
For any code $C\subseteq\{0,1\}^n$, define the function $f_C:\{0,1,*\}^n\rightarrow\{0,1,*\}^n$ to be
\begin{align}\label{deffi}
    f_C(x):=\begin{cases}
c & \text{if $c\in C$, $x_i\in\{c_i,*\}$ for all $i$, and $x$ has fewer than $d_{\min}(C)$ erasures}\\
(*,\dotsc,*) & \text{otherwise.}
\end{cases}
\end{align}
The function $f_C$ essentially tells us whether or not a partially-erased binary string with fewer than $d_{\min}(C)$ erasures is consistent with any codeword of $C$. It can always be computed in time $O(n^3)$ (see Note \ref{footnotetest}), and its runtime dictates the runtime of our following decoder for adversarial errors.
\begin{theorem}\label{thmadversarial}
    Consider any linear codes $C_1\subseteq\{0,1\}^{n_1},\dotsc,C_t\subseteq\{0,1\}^{n_t}$ and define $n:=\prod_{i=1}^tn_i  .$ Suppose there exists a function $T:\mathbb{N}\rightarrow \mathbb{N}$ such that $T(m)\geq m$ for all $m\in\mathbb{N}$ and such that for all $i\in[t]$, there is a $T(n_i)$-time algorithm for computing the function $f_{C_i}$ defined in (\ref{deffi}).
    Then there is an $O\left(\sum_{i=1}^t\frac{n}{n_i}\cdot T(n_i)\right)
$-time algorithm for decoding the tensor code  $C:=C_1\otimes\dotsc\otimes C_t$ from 
    \begin{align*}
        \left\lceil\frac{d_{\min}(C)}{2\max\big\{d_{\min}(C_1),\dots ,d_{\min}(C_t)\big\}}\right\rceil-1
    \end{align*}
    adversarial errors.
\end{theorem}

\begin{remark}\label{rmk1}
By Note \ref{footnotetest}, for any linear codes $C_1,\dotsc,C_t$ we can take $T(n)=O(n^3)$, which gives us a runtime of $O\left(n\sum_{i} n_i^2  \right)$. But for Reed-Muller codes, we can do better:
by Lemma \ref{rmtester}, we can take $T(n)=O(n\log n)$, which gives a runtime of $O\big(\sum_{i}n\log n_i\big)=O(n\log n)$. 
\end{remark}

\begin{remark}
Note that in Algorithm \ref{algotensor}, at each layer $i=2,\dotsc,t$ of the decoding process, the $i$-axis erasure pattern is the same for all $i$-axis vectors within the same $i$-subtensor of $A$. One could thus use Gaussian elimination to express the erased entries in this erasure pattern as a linear combination of the non-erased entries, then go through each $i$-axis vector of $A$ and correct the erasures accordingly.  If $C_1$ is taken to be the code of maximum length among $C_1,\dots,C_t$, this will give a running time of $O(n\sum_i n_i)$ for decoding the tensor product of any linear codes $C_1,\dotsc,C_t$ of lengths $n_1,\dotsc,n_t$. 
\end{remark}

\begin{proof}
    Note that we may assume that each $n_i$ is greater than $1$, as otherwise $C$ trivially reduces to a tensor product of $t-1$ codes.
    Our algorithm for the case where each $n_i$ is greater than $1$ is given in Algorithm \ref{algotensor}. 
\begin{figure}[ht]
\centering
\begin{minipage}{0.9\linewidth}
\begin{algorithm}[H]\label{algotensor}\caption{A polynomial-time decoder for arbitrary tensor codes}
\KwIn{An $n_1\times \dotsc \times n_t$ Boolean tensor $A$.} 

\KwOut{Either a codeword of the tensor code $C_1\otimes\dotsc\otimes C_t$ or an $n_1\times \dotsc \times n_t$ tensor filled with erasure symbols.}
    
    \If{$t=1$}{
    {If $A$ has fewer than $d_{\min}(C_1)$ erasures and there is a codeword $c\in C_1$ with $A_i\in\{c_i,*\}$ for all $i\in[n_1]$, replace $A$ by $c$. Otherwise replace every entry in $A$ by an erasure symbol.}}
    \Else{
        \For{$i=1,2,\dotsc,n_t$\label{linefor1}}{{Run Algorithm \ref{algotensor} on the $n_1\times\dotsc\times n_{t-1}$ tensor $A^i$ 
        whose entries are given by $A^i_{j_1\dots j_{t-1}}=A_{j_1\dots j_{t-1}i}$. Replace the entries $\{A_{j_1\dots j_{t-1}i}  \}$ of $A$ by
        the output entries.\label{defAi}}}
         \For{every $t-$axis vector $v\in\{0,1,*\}^{n_t}$ of $A$ 
        \label{linefor2} (see definition \ref{defnaxisvector})}{
        {If $v$ has fewer than $d_{\min}(C_t)$ erasures and there is a codeword $c\in C_t$ with $v_i\in\{c_i,*\}$ for all $i$, replace $v$ by $c$. Otherwise, replace $v$ by $(*,*,\dotsc,*)$.\label{linecorrecterasures}}
        
        \label{lineendcorrect}}
    \If{{the updated tensor $A$ contains erasure symbols or its Hamming distance from the original input is at least $\frac{d_{\min}(C_1)\cdot\dotsc \cdot d_{\min}(C_t)}{2}$\label{linechecks}\label{lineif}}}{Replace every entry of $A$ by an erasure symbol.\label{lineerasures}}
    }
    Output the updated tensor $A.$
\end{algorithm}
\end{minipage}
\end{figure}
We first show by induction that it always outputs either a codeword of $\mathsf{TRM}(r_1,m_1;\dotsc ;r_{t},m_{t})$ or the all-erasures tensor. 
The base case $t=1$ is trivial. For the case $t>1$, 
we note that by 
induction, after the line-\ref{linefor1} for loop completes, each $(t-1)$-subtensor of $A$ is either a valid codeword of $\mathsf{TRM}(r_1,m_1;\dotsc ;r_{t-1},m_{t-1})$ or a tensor filled with erasure symbols. Thus the erasure pattern of each $t$-axis vector in the line \ref{linefor2}-for loop is identical. In particular, if there is a unique consistent codeword for each of these $t$-axis vectors, then the erased entries of each $t$-axis vector can be expressed as the same linear combination of non-erased coordinates. This means that the erased $(t-1)$-subtensors can be expressed as linear combinations of the non-erased $(t-1)$-subtensors; since the code is linear, the newly recovered subtensors must then be codewords of $\mathsf{TRM}(r_1,m_1;\dotsc;r_{t-1},m_{t-1}).$ Combining this with the fact that by line \ref{linecorrecterasures}, each $t$-axis vector is a codeword of $\mathsf{RM}(r_t,m_t)$ (otherwise by line \ref{lineerasures}, we would output an all-erasures tensor), we get that any Boolean output must indeed be a codeword of $\mathsf{TRM}(r_1,m_1;\dotsc;r_t,m_t)$.

Now that we have proven that our algorithm always outputs either a valid codeword or a tensor filled with erasures, we turn to showing that the algorithm correctly ouputs the sent codeword as long as there are fewer than $\frac{d_{\min}(C)}{2\max\big\{d_{\min}(C_1),\dots ,d_{\min}(C_t)\big\}}$ errors. We again proceed by induction. 
    The base case $t=1$ is trivial. 
    For the general case, we note that in order for our Algorithm \ref{algotensor} to fail in decoding a noisy codeword containing fewer than $\frac{d_{\min}(C_1)\cdot\dotsc \cdot d_{\min}(C_{t})}{2}$ errors, one of the following two statements must hold:

    \begin{enumerate}[(i)]
        \item After we decode all the $(t-1)$-dimensional subtensors $\{A^i\}_{i\in[n_t]}$ (see line \ref{defAi}), there is a non-erasure erroneous entry in at least one of the updated subtensors $A^i$. \label{point1adv}
        \item After we decode all the $(t-1)$-dimensional subtensors $\{A^i\}_{i\in[n_t]}$, 
        there are at least $d_{\min}(C_t)$ values of $i\in[n_t]$ for which the updated subtensor $A^i$ contains one or more erasures.\label{point2adv}
    \end{enumerate}
    
Indeed, if neither of these occur, then our line \ref{linecorrecterasures} will allow us to recover every entry of $A$ correctly.
We now show that neither point (\ref{point1adv}) nor point (\ref{point2adv}) can occur.
Suppose for contradiction that there exists a subtensor $A^i$ as described in point (\ref{point1adv}). Note that by line \ref{lineerasures}, the Hamming distance between $A^i$ and its corresponding input must be less than $\frac{d_{\min}(C_1)\cdot\dotsc\cdot d_{\min}(C_{t-1})}{2}$. Since $A^i$ is a codeword of $C_1\times\dotsc\times C_{t-1}$ (we proved in the first paragraph that any output of our algorithm is a codeword) and since $C_1\times\dotsc\times C_{t-1}$ has minimum distance ${d_{\min}(C_1)\cdot\dotsc\cdot d_{\min}(C_{t-1})}$, there must have been at least $\frac{d_{\min}(C_1)\cdot\dotsc\cdot d_{\min}(C_{t-1})}{2}$ corrupted entries in the $i^\textnormal{th}$ subtensor to begin with. This contradicts our theorem's requirement on the total number of errors.

Suppose instead that there are $d_{\min}(C_t)$ subtensors $\{A^{i_1},\dotsc,A^{i_{d_{\min}(C_t)}}\}$ satisfying point (\ref{point2adv}) above. Since each of these subtensors is decoded independently, by our inductive hypothesis there must have been at least $d_{\min}(C_t)\cdot \frac{d_{\min}(C_1)\cdot\dotsc \cdot d_{\min}(C_{t-1})}{2\max\big\{d_{\min}(C_1),\dots ,d_{\min}(C_{t-1})\big\}}\geq \frac{d_{\min}(C)}{2\max\big\{d_{\min}(C_1),\dots ,d_{\min}(C_t)\big\}}$ errors. But this again contradicts our theorem's requirement on the total number of errors.

This concludes the proof of correctness. We now turn to the runtime analysis. 
Define $R(n_1,\dotsc,n_t)$ to be the maximal runtime of Algorithm \ref{algotensor} on any code $C'=C_1'\otimes\dotsc\otimes C_t'$ with $C_i'\subseteq \{0,1\}^{n_i}$ for all $i.$ 
Note that for $t>1$, we have
\begin{align*}
    R(n_1,\dotsc,n_t)&\leq n_t\cdot R(n_1,\dotsc, n_{t-1})+\frac{n}{n_t}T(n_t)+\alpha n,
\end{align*}
where the first and second terms correspond to the bulk of the runtime of Algorithm \ref{algotensor} for the for-loops \ref{linefor1} and \ref{linefor2} respectively, whereas the constant $\alpha$ is chosen to be big enough that the $\alpha n$-term covers all the other operations needed throughout the algorithm
. We claim that
\begin{align*}
    R(n_1,\dotsc,n_t)\leq (\alpha+1)\sum_{i=1}^t\frac{n}{n_i}\cdot T(n_i).
\end{align*}
For $t=1$, the statement is obvious. For $t>1$, by induction we have
\begin{align*}
     R(n_1,\dotsc,n_t)&\leq (\alpha +1)n_t\sum_{i=1}^{t-1}\frac{n/n_t}{n_i}T(n_i)+\frac{n}{n_t}T(n_t)+\alpha n\\
     &\leq (\alpha+1)\sum_{i=1}^{t}\frac{n}{n_i} T(n_i),
\end{align*}
 where in the last line we used the fact that by our theorem's requirement on $T$, we have $\alpha n=\frac{\alpha n}{n_t}n_t\leq \frac{\alpha n}{n_t}T(n_t)$.

\end{proof}

\section{Decoding Tensor Reed-Muller Codes from Random Errors}\label{sectiondecodeTRM}

In this section, we leverage our Algorithm \ref{algotensor} to prove the following formal version of Theorem \ref{mainthm}. Note that it is sufficient to prove Theorem \ref{mainthm} for $t\leq \sqrt{\log n}$, since after that the $O\left(\frac{1}{t}\right)$ improvement in the error probability is subsumed by the $o(1)$ term (one can always artificially increase $t$ by taking tensor products with the trivial Reed-Muller code $\{0,1\}$).

\begin{theorem}\label{mainthmformal}
Consider any constants $p\in(0,\frac{1}2{})$ and $R<1-h(p)$. Let $n\in\mathbb{N}$ be some growing parameter and consider any corresponding integer $3\leq t\leq \sqrt{\log n}$. Define
\begin{itemize}
    \item $m_1:=\lceil\log\log n-3\rceil$ and $r_1$ to be any integer such that $\frac{  \binom{m_1}{\leq r_1}}{2^{m_1}}=R\pm o(1)$
    \item $m_2:=\lceil10\log\log n\rceil$ and $ r_2:=\lceil\frac{m_2}{2}+\sqrt{m_2}\log m_2\rceil$
    \item $ m_3=\dotsc=m_t:=\lceil\frac{\log n-m_1-m_2}{t-2}\rceil$ and $ r_3=\dotsc=r_t:=\lceil\frac{m_3+m_3^{3/4}}{2}\rceil$ 
\end{itemize}
Then the Tensor Reed-Muller code $\mathsf{TRM}(r_1,m_1;\dotsc;r_t,m_t)$ 
has length $n^{1+ o(n)}$ and rate $R\pm o(1)$. Moreover, there exists a decoder $D$ 
with the following properties:
\begin{enumerate}
    \item $D$ has runtime $O\left(n\log{\log n}\right)$ in the case $t=3$ and runtime $O\left(n\log n\right)$ in the case $t>3.$
    
    \item For every codeword $c\in \mathsf{TRM}(r_1,m_1;\dotsc;r_t,m_t)$, $D$ has decoding failure probability
\begin{align*}
        \Pr_{z\sim p}\Big[D(c+z)\neq c\Big]
    \leq \begin{cases}
    n^{-\omega(\log n)}
    & \text{if $t=3$,}\\
2^{-n^{\frac{1}{2}-\frac{1}{2(t-2)}-o(1)}} & \text{otherwise.}
\end{cases}
    \end{align*}


\end{enumerate}

\end{theorem}

\begin{proof}
    Note that $\mathsf{TRM}(r_1,m_1;\dotsc;r_t,m_t)$ has rate 
    \begin{align*}
        \prod_{i=1}^t\frac{\binom{m_i}{\leq r_i}}{2^{m_i}}
        &\geq \big(R- o(1)\big)\prod_{i=3}^{t}\left( 1-\frac{2^{h(\frac{m_i-r_i}{m_i})m_i}}{2^{m_i}} \right)\\
        &\geq (R-o(1))\prod_{i=3}^{t}\left( 1-2^{-\frac{\sqrt{m_3}}{2\ln2}} \right)\\
        &\geq (R-o(1))(1-t2^{-\frac{\sqrt{m_3}}{2\ln2}} ),\\
    \end{align*}
    where the first inequality follows from the fact that $\binom{n}{\leq d}\leq 2^{h(\frac{d}{n})n}$ for all integers $n$ and $d\leq \frac{n}{2}$, the second inequality follows from the fact that $h\left(\frac{1-\mu}{2}\right)\leq 1-\frac{\mu^2}{2\ln2}$ for any $\mu\in(0,1)$, and the third inequality follows from the fact that $(1+x)^r\geq 1+rx$ for all $x\geq -1$ and $r\geq1.$ Since $t\leq\sqrt{\log n}$ and $m_3=\Omega(\frac{\log n}{t})$, our code $\mathsf{TRM}(r_1,m_1;\dotsc;r_t,m_t)$ indeed has rate $R\pm o(1).$ It also has length $2^{m_1+m_2}\cdot 2^{
    (t-2)\lceil\frac{\log n-m_1-m_2}{t-2}\rceil}=n^{1+o(1)}.$

    The decoder $D$ we use for decoding our code is given in Algorithm \ref{algotrm}. We represent each codeword of $\mathsf{TRM}(r_1,m_1;\dotsc;r_t,m_t)$ as a $2^{m_1}\times\dotsc\times 2^{m_t}$ Boolean tensor. For any tensor $A\in \{0,1\}^{2^{m_1}\times \dotsc\times2^{m_t} }$, we call any vector along the first axis of $A$ a "row" of $A$ and call any vector along the second axis of $A$ a "column" of $A$. 
\begin{figure}[ht]
\centering
\begin{minipage}{0.9\linewidth}
\begin{algorithm}[H]\label{algotrm}\caption{An efficient decoder for $t$-tensor Reed-Muller codes}
\KwIn{A $2^{m_1}\times \dotsc\times 2^{m_t}$ Boolean tensor $A$.}
\KwOut{
A $2^{m_1}\times \dotsc\times 2^{m_t}$ Boolean tensor. 
}

    Create a look-up table $D_1:\{0,1\}^{2^{m_1}}\rightarrow\mathsf{RM}(r_1,m_1).$
    
    \For{each vector $s\in\{0,1\}^{2^{m_1}}$}{Use brute-force search to find the $c\in\mathsf{RM}(r_1,m_1)$ closest to $s.$ Set $D_1[s]=c.$
    }
    
	\For{each row $u$ of $A$\label{rowforloop}}{Replace $u$ by $D_1[u].$\label{insiderowforloop}}\label{endrowforloop}
    
    $\mathsf{counter} \leftarrow 0.$
    
    \For{each column $v$ of our updated tensor $A$\label{columnforloop}}{{Use Lemma \ref{rmtester} to check if $v\in\mathsf{RM}(r_2,m_2)$. If not, increase $\mathsf{counter}$ by $1$ and replace the column $v$ by the codeword $D_2(v)\in\mathsf{RM}(r_2,m_2)$, where $D_2$ is the decoder from Lemma \ref{combinedecoders} for the code $\mathsf{RM}(r_2,m_2)$. If $\mathsf{counter}>n2^{-2^{(\log\log n)^{1/4}}}$, abort the entire algorithm and return the $\vec{0}$ codeword.\label{linecountert}}}\label{endcolumnforloop}
 {If $t=3$, output the updated tensor $A.$ If $t>3$, run Algorithm \ref{algotensor} on $A$ and return the output 
 (if Algorithm \ref{algotensor} outputs a tensor filled with erasure symbols, return the 
 $\vec{0}$ codeword).\label{lastline}}
\end{algorithm}
\end{minipage}
\end{figure}
We first bound the probability that Algorithm \ref{algotrm} makes a decoding mistake. Note that the only way a decoding mistake can occur is if either:
\begin{enumerate}[(i)]
\item Our algorithm would fail even if the $\mathsf{counter}$ condition (last sentence of line \ref{linecountert}) was disregarded.\label{disregardcounter}
    \item The $\mathsf{counter}$ eventually exceeds $n2^{-2^{(\log\log n)^{1/4}}}$.\label{bigcounter}
\end{enumerate}
We first show that point (\ref{disregardcounter}) is very unlikely to occur. Note that by Theorem \ref{errorcapacity}, the lookup table $D_{1}$ used in the row-for loop of our Algorithm \ref{algotrm} (line \ref{insiderowforloop}) satisfies that for any $c_1\in \mathsf{RM}(r_1,m_1)$,
\begin{align}\label{probx}
     \Pr_{z\sim p}\Big[ D_{1}[c_{1}+z]\neq c_1 \Big]&\leq 2^{-2^{\Omega(\sqrt{m_1})}}\nonumber\\
     &\leq
     2^{-2^{\Omega(\sqrt{\log\log n})}}.
\end{align}
Thus at the end of the row-for loop, every entry of $A$ has probability $\varepsilon\leq 2^{-2^{\Omega(\sqrt{\log\log n})}}$ of being different from the corresponding entry of the sent codeword $B$. Furthermore, since $D_1$ was applied independently to every row of $A$, any coordinates of $A$ that are not in the same row have uncorrelated probabilities of being correct. 

In particular, at the end of the row-for loop, the entries of any given column of $A$ have uncorrelated probabilities of being incorrect. Now, since $\varepsilon\leq\frac{2^{-m_2}}{5}$ for all $n$ large enough, we get from Lemma \ref{combinedecoders} that for any $c_2\in\mathsf{RM}(r_2,m_2)$, the decoder $D_2$ used in the column-for loop (line \ref{linecountert}) of our Algorithm \ref{algotrm} satisfies
\begin{align}\label{proby}
    \Pr_{z\sim \varepsilon}\Big[ D_2(c_2+z)\neq c_2 \Big]&\leq O\left(2^{-2^{\frac{m_2-r_2}{2}}}\right)\nonumber\\
    &=
    2^{-2^{2.5\log\log n\pm o(\log\log n)}}\nonumber\\
    &\leq 2^{-\omega\left(\log^2 n\right)}
\end{align}
Thus at the end of our column-for loop, if we disregard the $\mathsf{counter}$ condition when running the algorithm, then every entry of $A_{}$ has probability 
\begin{align}\label{probaftercolloop}
    \varepsilon'\leq n^{-\omega(\log n)}
\end{align}
of being incorrect. Taking a union bound over all coordinates then concludes the analysis of point (\ref{disregardcounter}) for the case $t=3.$ For the case $t>3$, we define for every $k\in[2^{m_3}]\times\dotsc\times [2^{m_t}]$ the Boolean random variable $X_k$ that is $1$ if and only if upon running Algorithm \ref{algotrm} without the $\mathsf{counter}$ condition, at the end of the column-for loop (line \ref{endcolumnforloop}), there exists $(i,j)\in[2^{m_1}]\times [2^{m_2}]$ such that the entry $A_{ijk}$ is incorrect. By (\ref{probaftercolloop}), since $D_1$ was applied independently to every row and $D_2$ was applied independently to every column, 
the random variables $\{X_k\}$ are independent Bernoulli random variables with probability parameter at most $ 2^{m_1+m_2}\varepsilon'=n^{-\omega(\log n)}.$ By the Chernoff bound (Lemma \ref{chernoff}), 
we then have

\begin{align*}
    \Pr\left[ \sum_{k}X_k\geq {\frac{d_{\min}\big(\mathsf{TRM}(r_1,m_1;\dotsc;r_t,m_t)\big)}{2^{m_1+m_2+1}d_{\min}\big(\mathsf{RM}(r_3,m_3)\big)}}\right]
    &\leq 
    2^{-\frac{d_{\min}\big(\mathsf{TRM}(r_1,m_1;\dotsc;r_t,m_t)\big)}{d_{\min}\big(\mathsf{RM}(r_3,m_3)\big)}\cdot n^{-o(1)}}\nonumber\\
    &\leq 
    2^{-n^{\frac{t-3}{2(t-2)}-o(1)}}.
\end{align*}
Thus if we disregard the $\mathsf{counter}$ condition when running the algorithm, we get that with probability $1-2^{-n^{\frac{t-3}{2(t-2)}-o(1)}}$, there are $\leq\frac{d_{\min}\big(\mathsf{TRM}(r_1,m_1;\dotsc;r_t,m_t)\big)}{2d_{\min}\big(\mathsf{RM}(r_3,m_3)\big)}$ errors remaining after the line-\ref{columnforloop} for loop. By Theorem \ref{thmadversarial}, line \ref{lastline} will then succeed with probability at least $1-2^{-n^{\frac{1}{2}-\frac{1}{2(t-2)}-o(1)}}.$ This concludes our analysis for point (\ref{disregardcounter}). We then turn to showing that point (\ref{bigcounter}) is very unlikely to occur.
Note that by (\ref{probx}), at the end of the row-for loop (line \ref{endrowforloop}) of our algorithm, for any $k\in[2^{m_3}]\times\dotsc\times [2^{m_t}]$ we have
\begin{align*}
    \Pr\Big[\exists (i,j)\in[2^{m_1}]\times[2^{m_2}]\textnormal{ such that }A_{i,j,k}\textnormal{ is incorrect}\Big]&\leq 2^{m_2}\cdot 2^{-2^{\Omega(\sqrt{\log\log n})}}\\
    &\leq 2^{-2^{\Omega(\sqrt{\log\log n})}}.
\end{align*}
Since Algorithm \ref{algotrm} processes each $k\in[2^{m_3}]\times\dotsc\times [2^{m_t}]$ independently up to the end of the column-for loop (line \ref{endcolumnforloop}), we get
\begin{align*}
    \Pr\left[\mathsf{counter}\textnormal{ exceeds }n2^{-2^{(\log\log n)^{1/4}}}\right]&\leq \Pr\left[\sum_{k=1}^{n2^{-m_1-m_2}}2^{m_1}Y_k\geq n2^{-2^{(\log\log n)^{1/4}}} \right]
\end{align*}
for $\{Y_k\}$ independent Bernoulli variables of probability $2^{-2^{\Omega(\sqrt{\log\log n})}}$. By the Chernoff bound (Lemma \ref{chernoff}), we then have
\begin{align*}
    \Pr\left[\mathsf{counter}\textnormal{ exceeds }n2^{-2^{(\log\log n)^{1/4}}}\right]&\leq 2^{-\Omega\left(n2^{-2^{(\log\log n)^{1/4}}}/2^{m_1}\right)}\leq O\left(2^{-\sqrt{n}}\right).
\end{align*}
Combining this bound for point (\ref{bigcounter}) with our previously established bound for point (\ref{disregardcounter}), we get
\begin{align*}
    \Pr_{z\sim p}\Big[ D(c+z)\neq c \Big]&\leq n^{-\omega(\log n)}+O\left(2^{-\sqrt{n}}\right)
\end{align*}
for the case $t=3$ and
\begin{align*}
    \Pr_{z\sim p}\Big[ D(c+z)\neq c \Big]&\leq 2^{-n^{\frac{1}{2}-\frac{1}{2(t-2)}-o(1)}}+O\left(2^{-\sqrt{n}}\right)
\end{align*}
for the case $t>3$, as desired. We now turn to bounding our algorithm's runtime. Since there are $2^{2^{m_1}}$ vectors in $\{0,1\}^{2^{m_1}}$, creating the look-up table $D_1$ takes time 
\begin{align}\label{maketable}
2^{2^{m_1}}\cdot O(2^{2^{m_1}}2^{m_1})=o(n).
\end{align}
Since there are $\frac{n}{2^{m_1}}$ rows in the tensor $A$ and since looking up a value in the table $D_1$ takes time $O(2^{m_1})$, the row-for loop (line \ref{rowforloop}) then takes time 
\begin{align}\label{runtimebd1}
\frac{n}{2^{m_1}}\cdot O\left(2^{m_1}\right)=O\left(n\right).
\end{align}
For the column-for loop (line \ref{columnforloop}), since there are $\frac{n}{2^{m_2}}$ columns in the tensor $A$, by Lemmas \ref{rmtester} and \ref{combinedecoders} the algorithm takes time
\begin{align}\label{runtimebd2}
    \frac{n}{2^{m_2}}\cdot O(m_22^{m_2})+n2^{-2^{(\log\log n)^{1/4}}}\cdot 2^{O(m_2)}=O\left(n\log{\log n}\right).
\end{align}
Combining equations (\ref{maketable}), (\ref{runtimebd1}) and (\ref{runtimebd2}), we get that our decoder $D$ has total runtime $O\left(n\log{\log n}\right)$ when $t=3.$ When $t>3$, the algorithm additionally has to process line \ref{lastline}, which takes time \begin{align}\label{runtimebd3}
        O\left(n\log n\right)
\end{align}
by Theorem \ref{thmadversarial} and Lemma \ref{rmtester}. This brings the total runtime for the case $t>3$ to $O\left(n\log n\right).$

\end{proof}

\bibliographystyle{alpha}
\bibliography{tensor}

\end{document}